\documentclass{article}

\usepackage{amsmath,amstext,amssymb,amsfonts}
\usepackage{graphicx}
\usepackage[dvipsnames]{xcolor}
\usepackage{tabularx}
\usepackage{verbatim}
\usepackage{fullpage}
\usepackage{bbm}

\usepackage{amsthm}
\usepackage{ifdraft}
\usepackage{algorithm2e}
\usepackage{algorithmic}
\usepackage{appendix}
\usepackage{enumerate}
\usepackage{multicol}

\usepackage{nicefrac}

\newcommand{\flatfrac}[2]{#1/#2}
\newcommand{\ffrac}{\flatfrac}

\usepackage{microtype}

\newtheorem{theorem}{Theorem}[section]
\newtheorem*{theorem*}{Theorem}

\newtheorem{Claim}[theorem]{Claim}
\newtheorem*{claim*}{Claim}

\newtheorem{proposition}[theorem]{Proposition}
\newtheorem*{proposition*}{Proposition}
\newtheorem{lemma}[theorem]{Lemma}
\newtheorem*{lemma*}{Lemma}
\newtheorem{corollary}[theorem]{Corollary}

\newtheorem*{conjecture*}{Conjecture}

\newtheorem{fact}[theorem]{Fact}
\newtheorem*{fact*}{Fact}

\newtheorem*{hypothesis*}{Hypothesis}

\theoremstyle{definition}
\newtheorem{definition}[theorem]{Definition}
\newtheorem{construction}[theorem]{Construction}

\newtheorem{SDP}[theorem]{SDP}

\usepackage{prettyref}
\newcommand{\savehyperref}[2]{\texorpdfstring{\hyperref[#1]{#2}}{#2}}

\newrefformat{parta}{\savehyperref{#1}{a}}
\newrefformat{partb}{\savehyperref{#1}{b}}
\newrefformat{eq}{\savehyperref{#1}{\textup{(\ref*{#1})}}}
\newrefformat{lem}{\savehyperref{#1}{Lemma~\ref*{#1}}}
\newrefformat{def}{\savehyperref{#1}{Definition~\ref*{#1}}}
\newrefformat{thm}{\savehyperref{#1}{Theorem~\ref*{#1}}}
\newrefformat{cor}{\savehyperref{#1}{Corollary~\ref*{#1}}}
\newrefformat{cha}{\savehyperref{#1}{Chapter~\ref*{#1}}}
\newrefformat{sec}{\savehyperref{#1}{Section~\ref*{#1}}}
\newrefformat{app}{\savehyperref{#1}{Appendix~\ref*{#1}}}
\newrefformat{tab}{\savehyperref{#1}{Table~\ref*{#1}}}
\newrefformat{fig}{\savehyperref{#1}{Figure~\ref*{#1}}}
\newrefformat{hyp}{\savehyperref{#1}{Hypothesis~\ref*{#1}}}
\newrefformat{alg}{\savehyperref{#1}{Algorithm~\ref*{#1}}}
\newrefformat{sdp}{\savehyperref{#1}{SDP~\ref*{#1}}}
\newrefformat{vp}{\savehyperref{#1}{VP~\ref*{#1}}}
\newrefformat{lp}{\savehyperref{#1}{LP~\ref*{#1}}}
\newrefformat{rem}{\savehyperref{#1}{Remark~\ref*{#1}}}
\newrefformat{item}{\savehyperref{#1}{Item~\ref*{#1}}}
\newrefformat{step}{\savehyperref{#1}{step~\ref*{#1}}}
\newrefformat{conj}{\savehyperref{#1}{Conjecture~\ref*{#1}}}
\newrefformat{fact}{\savehyperref{#1}{Fact~\ref*{#1}}}
\newrefformat{prop}{\savehyperref{#1}{Proposition~\ref*{#1}}}
\newrefformat{claim}{\savehyperref{#1}{Claim~\ref*{#1}}}
\newrefformat{relax}{\savehyperref{#1}{Relaxation~\ref*{#1}}}
\newrefformat{red}{\savehyperref{#1}{Reduction~\ref*{#1}}}
\newrefformat{part}{\savehyperref{#1}{Part~\ref*{#1}}}
\newrefformat{prob}{\savehyperref{#1}{Problem~\ref*{#1}}}
\newrefformat{ass}{\savehyperref{#1}{Assumption~\ref*{#1}}}
\newrefformat{cons}{\savehyperref{#1}{Construction~\ref*{#1}}}

\newcommand{\Sref}[1]{\hyperref[#1]{\S\ref*{#1}}}

\renewcommand{\leq}{\leqslant}

\renewcommand{\geq}{\geqslant}

\usepackage{bm}

\usepackage{xspace}

\usepackage[pdftex,pagebackref,colorlinks,linkcolor=blue,filecolor = blue, citecolor = magenta, urlcolor  = JungleGreen]{hyperref}

\usepackage{boxedminipage}

\newenvironment{mybox}
{\center \noindent\begin{boxedminipage}{1.0\linewidth}}
{\end{boxedminipage}
\noindent
}

\newcommand{\mper}{\,.}

\newcommand{\paren}[1]{\left(#1 \right )}

\newcommand{\brac}[1]{[#1 ]}

\newcommand{\set}[1]{\left\{#1\right\}}

\newcommand{\abs}[1]{\left\lvert#1\right\rvert}
\newcommand{\Abs}[1]{\left\lvert#1\right\rvert}

\newcommand{\norm}[1]{\left\lVert#1\right\rVert}

\newcommand{\defeq}{\stackrel{\textup{def}}{=}}

\newcommand{\inprod}[1]{\left\langle #1\right\rangle}

\newcommand{\R}{\mathbb R}

\newcommand{\subjectto}{\text{subject to}}

\newcommand{\Esymb}{\mathbb{E}}
\newcommand{\Psymb}{\mathbb{P}}

\DeclareMathOperator*{\E}{\Esymb}
\DeclareMathOperator*{\Var}{{\sf Var}}

\DeclareMathOperator*{\ProbOp}{\Psymb}
\newcommand{\var}[1]{\Var \left[#1\right]}

\newcommand{\e}{\epsilon}

\newcommand{\one}{\mathbbm{1}}

\definecolor{DSgray}{cmyk}{0,0,0,0.7}

\let\e\varepsilon

\newcommand{\Lovasz}{Lov\'asz\xspace}

\newcommand{\Hastad}{H{\aa}stad\xspace}

\newcommand{\etal}{et al. }
\newcommand{\bigO}{\mathcal{O}}
\newcommand{\bigo}[1]{\bigO\left(#1\right)}

\newcommand{\poly}{{\sf poly}}

\newcommand\numberthis{\addtocounter{equation}{1}\tag{\theequation}} 

\author{
	Yash Khanna \\IISc, Bangalore\footnote{Indian Institute of Science, Bangalore, India.}\\ \href{mailto:yashkhanna846@gmail.com}{yashkhanna846@gmail.com}
	\and 
	Anand Louis\\IISc, Bangalore\footnotemark[1]\\ \href{mailto:anandl@iisc.ac.in}{anandl@iisc.ac.in}
	\and 
	Rameesh Paul\\IISc, Bangalore\footnotemark[1]\\ \href{mailto:rameeshpaul@iisc.ac.in}{rameeshpaul@iisc.ac.in}
}

\title{Independent Sets in Semi-random Hypergraphs}

\begin{document}
\maketitle

\begin{abstract}
A set of vertices in a hypergraph is called an independent set if no hyperedge is completely contained inside the set. Given a hypergraph, computing its largest size independent set is an NP-hard problem. 

In this work, we study the independent set problem on hypergraphs in a natural semi-random family of instances. 
Our semi-random model is inspired by the Feige-Kilian model \cite{MR1894527}. This popular model has also been studied in the works of \cite{MR1894527, DBLP:journals/eccc/Steinhardt17, McKenzie_2020} etc. McKenzie, Mehta, and Trevisan \cite{McKenzie_2020} gave algorithms for computing independent sets in such a semi-random family of graphs. The algorithms by McKenzie~\etal \cite{McKenzie_2020} are based on rounding a ``crude-SDP''. We generalize their results and techniques to hypergraphs for an analogous family of hypergraph instances. Our algorithms are based on rounding the ``crude-SDP'' of McKenzie~\etal \cite{McKenzie_2020}, augmented with ``Lasserre/SoS like'' hierarchy of constraints. Analogous to the results of McKenzie~\etal \cite{McKenzie_2020}, we study the ranges of input parameters where we can recover the planted independent set or a large independent set.
\end{abstract}

\section{Introduction}
An independent set of a hypergraph $H=(V,E)$ is a subset of vertices such that no hyperedge is completely contained inside the subset.
Computing a maximum independent set is a fundamental problem in the study of algorithms. The problem has applications in areas such as resource allocation in wireless networks \cite{MR3700397}, data clustering \cite{6816636}, computational biology \cite{MR2516078}, etc.

The problem of computing a maximum size independent set in graphs is well known to be NP-hard \cite{Karp1972}.
\Hastad \cite{DBLP:journals/eccc/ECCC-TR97-038} showed that it is hard to approximate the maximum independent set in graphs to better than a factor of $n^{1-\e}$ for any $\e >0$ unless $NP=ZPP$. Zuckerman \cite{v003a006} showed that there is no possible approximation ratio better than $n^{1-\e}$ unless $P=NP$. This hardness of approximation holds for the independent set problem on hypergraphs as well, since it generalizes the independent set problem on graphs. 

There has been a lot of work studying approximation algorithms of independent sets in graphs and hypergraphs, see \prettyref{sec:related_work} for a brief survey. Another direction of research related to intractable problems is to study families of ``easier'' instances of the problem. This includes studying various random and semi-random models of instances, instances satisfying certain properties, etc. We give a brief survey of the special class of graphs for which the independent set problem has been studied in \prettyref{sec:related_work}.

The starting point in the study of random instances for the independent set problem in graphs were the $G(n,p)$ instances (Erdős–Rényi random graphs). 
Analysis of $G(n,p)$ \cite{matula} showed that a random graph don't have an independent set of size more than $\paren{2+ o(1)}\log_{1/(1-p)} n$, w.h.p., for a large range of $p$. A simple algorithm can be used to compute an independent set of size $\log_{1/(1-p)} n$, w.h.p., but computing an independent set larger than this seems to be hard. Another popular model, the {\em planted solution model} considers the problem of recovering a hidden planted structure of size $k$ in a graph with $n$ vertices. 
For the planted clique (or independent set) model, Alon, Krivelevich and Sudakov \cite{MR1662795} showed that we can recover the planted clique as long as $k=\Omega(\sqrt{n})$ (for a constant $p$). 
Blum and Spencer \cite{BLUM1995204} studied {\em semi-random models} of $k$-colorable graphs; such models allow an adversary to modify the instance without changing the planted structure. The model is defined by the set of actions allowed to the adversary.
For the planted clique problem, in a rich adversarial semi-random model introduced by Feige and Kilian \cite{MR1894527}, the algorithm of \cite{McKenzie_2020} can recover the clique for $k = \Omega_p(n^{\ffrac{2}{3}})$ (see \prettyref{sec:related_work} for precise statement).
We note that algorithms for independent set in graphs mentioned above hold more generally; the results here are stated assuming a constant value of $p$ for the purpose of illustration.

The above is a broad classifications of the models and there are many other probabilistic generative models that fit at some intermediate hierarchy in this classification. A key advantage of studying random and semi-random instances is that it gives us insights into which aspects of the problem makes it hard. 
Often algorithms for stronger models and stronger regimes of parameters may require using more advanced tools and techniques. For example, in the case of planted cliques/independent sets, for the regimes of $k\geq \Omega\paren{\sqrt{n\log n}}$ we can recover the planted graph using combinatorial techniques which essentially returns the vertices with top $k$ degrees. However in regimes of $k = \Omega\paren{\sqrt{n}}$ this approach no longer works, and the best known algorithms \cite{MR1662795} use spectral techniques. In the regime of semi-random instances of the problem, the best known algorithms \cite{McKenzie_2020} are based on semidefinite programming. 

\subsection{Our Models and Results}

\begin{definition}
\label{def:model}
Given parameters $n, k, r$, and $p$, a hypergraph $H$ is constructed as follows. 
\begin{enumerate}
\item Let $V$ be a set of $n$ vertices. Fix an arbitrary subset $S \subset V$ of size $k$.
\item Add a hyperedge independently with probability $p$ for each $r-$tuple of vertices $\set{i_1,i_2,\hdots,i_r}$, such that $\set{i_1,i_2,\hdots,i_r} \cap S \neq \emptyset$ and $\set{i_1,i_2,\hdots,i_r} \cap (V \setminus S) \neq \emptyset$. We denote the hypergraph induced by collection of such $r$-tuples as $H[S, V \setminus S]$.
\item Arbitrarily add $r$-hyperedges to the set $V \setminus S$.
\item Allow a \emph{monotone adversary} to add $r$-hyperedges arbitrarily to $H[S, V \setminus S]$ and hypergraph induced on $V \setminus S$ denoted by $H[V \setminus S]$.
\end{enumerate}
\end{definition}
The model discussed above was introduced by Feige and Kilian \cite{MR1894527} in the context of studying various graph partitioning problems.
The work \cite{McKenzie_2020} studied an analogous model in the context of independent sets in graphs. 

We study the ranges of parameters $k, r, p$ (for a fixed $n$) in this model for which we can recover $S$ efficiently. Our main results are informally stated below.	
\begin{theorem}[Informal version of \prettyref{thm:main_formal}]
	\label{thm:main}
	There exists a deterministic algorithm which takes as input 	
	an instance of \prettyref{def:model} satisfying
	\[ k=\Omega \paren{\dfrac{n^{(r-1)/(r-0.5)}}{p^{3/(2r-1)}}}, \]
	has running time $n^{\bigo{r}}$, and outputs a list of atmost $n$ independent sets, one of which is $S$,  with high probability (over the randomness of the input).
\end{theorem}
\begin{theorem}[Informal version of \prettyref{thm:main2_formal}]
	\label{thm:main2}
	There exists a deterministic algorithm which takes as input $\e \in (0,1)$ and	
	an instance of \prettyref{def:model} satisfying 
    \[ k=\Omega{\paren{\dfrac{n^{(r-1)/(r-0.5)}}{\e^{1/(r-0.5)}p^{1/(2r-1)}}}}, \]
	has running time $n^{\bigo{r}}$, and outputs an independent set of size at least $(1 - \e)k$, with high probability (over the randomness of the input).
\end{theorem}
\prettyref{thm:main} and \prettyref{thm:main2} generalize to hypergraphs the analogous results for graphs by \cite{McKenzie_2020}. We state and prove formal version of these results in \prettyref{thm:main_formal} and \prettyref{thm:main2_formal} respectively.
Our proofs of \prettyref{thm:main_formal} and \prettyref{thm:main2_formal}  are based on rounding  McKenzie~\etal \cite{McKenzie_2020} ``crude-SDP'', augmented with ``Lasserre/SoS like'' hierarchy of constraints. The Lasserre/SoS hierarchy has been used in designing approximation algorithms for independent sets in hypergraphs in the works by Chlamtac \cite{4389537} and Chlamtac and Singh \cite{MR2538776}, but the power of the Lasserre/SoS hierarchy for designing approximation algorithm for independent set problem is yet to be fully understood.

\subsection{Related Work}
\label{sec:related_work}
\paragraph{Independent set problem in hypergraphs.}
The independent set problem in hypergraphs cannot be approximated to a factor better than $n^{1-\e}$ for any $\e >0$ unless P=NP \cite{v003a006}. The work \cite{MR1684102} gives a combinatorial algorithm to obtain an approximation ratio of $\bigO \paren{n/\paren{\log ^{\paren{r-1}}n}^2}$ for a $r$-uniform hypergraph where $\log^{\paren{r}}n$ denotes a $r$-fold repeated application of logarithm as $\log \hdots \log n$. This has been improved by Halld\'{o}rsson in the work \cite{MR1783779} where they study the problem on arbitrary weighted hypergraphs and give a $\bigO \paren{n/ \log n}$ approximation algorithm that runs in $\poly \paren{n,m}$ time where $m$ denotes the number of hyperedges. From here onwards a lot of work has been done in studying the problem in special class of graphs. In this section we do a brief survey of these results.

The problem has been extensively studied for $3$-uniform hypergraphs which contain an independent set of size $\gamma n$. Krivelevich, Nathaniel and Sudakov \cite{MR1855351} give an SDP based algorithm that finds an independent set of size $\tilde{\Omega}\paren{\min{\paren{n,n^{6\gamma-3}}}}$ for $\gamma \geq \ffrac{1}{2}$.
The work Chlamtac \cite{4389537} uses a SDP relaxation with the third level of the Lasserre/SoS hierarchy
and returns an independent set of size $\Omega\paren{n^{1/2 - \gamma}}$. Chlamtac and Singh \cite{MR2538776} gave an algorithm which computes an independent set of size $n^{\Omega(\gamma^2)}$ (where $\gamma \geq 0$ is a constant) using $\Theta(\ffrac{1}{\gamma^2})$ levels of a mixed hierarchy which they called {\em the intermediate hierarchy}. 
 The Lasserre hierarchy has been used in designing approximation algorithms for various problems \cite{BRS11, GS11,AJT19}, etc. 

Halld\'{o}rsson and Losievskaja \cite{MR2514594} study the problem on bounded degree hypergraphs. For hypergraphs with degree bounded by $\Delta$, they show that the classical greedy set cover algorithm can be analyzed to give $\ffrac{\paren{\Delta+1}}{2}$ approximation. The work \cite{MR2804389} shows that the bounded degree case is Unique Games-hard to approximate within a factor of $\bigO \paren{\ffrac{\Delta}{\log^2 \Delta}}$. In a recent work \cite{bhangale_et_al:LIPIcs:2019:10825}, they exhibit how to convert this inapproximability factor of $\bigO \paren{\ffrac{\Delta}{\log^2 \Delta}}$ under UG-hardness to NP-hardness.  

\paragraph{Random models for Independent set problem.}
The model studied in this work is a generalization (to hypergraphs) of the planted independent set model on graphs studied in \cite{McKenzie_2020}. Their algorithm is based on rounding a SDP solution. However, instead of using a relaxation of the independent set problem, they used a crude-SDP (this idea was introduced in \cite{Kolla_2011} and also used in many subsequent works \cite{10.1145/2213977.2214013}) which helps reveal the planted solution $S$. The main idea is to show that the expected $\ell_2^2$ distance between vectors of $S$ (the planted independent set) is ``small". In other words the SDP solution ``clusters" the vectors of $S$. Their algorithm outputs an independent set of size $(1-\e)k$ for $k = \Omega\paren{\flatfrac{n^{2/3}}{p^{1/3}}}$ and for a larger value of $k$, i.e. when $k = \Omega\paren{\flatfrac{n^{2/3}}{p}}$, it outputs atmost $n$ independent sets, one of which is the planted one w.h.p. In this parameter range, they also consider a list decoding version, where when given a random vertex of $S$ correctly picks $S$ from this list. The proofs of \prettyref{thm:main2} and \prettyref{thm:main} generalize the proofs of the corresponding results in \cite{McKenzie_2020}.

The problem has also been studied in graphs in a weaker semi-random model \cite{DBLP:journals/rsa/FeigeK00} by Feige and Krauthgamer which they call as the sandwich model. They propose an algorithm based on \Lovasz theta function for the same which returns the planted clique for $k \geq \Omega(\sqrt{n})$ (for $p = 1/2$). Feige and Kilian \cite{MR1894527} studied the problem in their semi-random model and they give an algorithm to recover an independent set of size $\alpha n$ for regimes of $p > \ffrac{(1+\e)\ln n}{\alpha n}$ and any $\e > 0$, where $\alpha$ is a constant. They also give efficient algorithms to recover a planted bisection and planted $k$-colorable graphs in semi-random models.

A closely related problem is about recovering planted clusters in random graphs known as the Stochastic Block Model (SBM) given by \cite{MR718088}. In \cite{MR4066151} they study the hypergraph version of the problem where they partition a $r$-uniform random hypergraph $H(n,r,p,q)$ into $k$ equally sized clusters with $p$ as edge probability within a cluster and $q$ as edge probability amongst clusters. They give a spectral algorithm which guarantees exact recovery when number of clusters $k=\Theta(\sqrt{n})$. The work \cite{ghoshdastidar2017} studies this problem in more general models like the planted partition model for non-uniform hypergraphs. The work \cite{DBLP:journals/corr/abs-1807-02884} gives an SDP based algorithm for the community detection problem in $k$-uniform hypergraphs.

\paragraph{Other problems in semi-random models.}
In \cite{10.1145/2213977.2214013} they develop a general framework to study graph partition problems in a semi-random model similar (in strength) to the one by Feige and Kilian \cite{MR1894527}. They give bi-criteria approximation algorithms for Sparsest cut, Uncut, Multi cut, Balanced Cut and Small set expansion problems. In \cite{MR3238929} they propose another semi-random model which they call as PIE (permutation invariant edges model) for the balanced cut problem.  The works by Khanna, Louis, and Venkat \cite{DBLP:conf/icalp/LouisV18, MR4042189, khanna_et_al:LIPIcs:2020:13268} study the problems of graph expansion (vertex and edge), and the densest $k$-subgraph problem in semi-random models. The work by Khanna \cite{DBLP:journals/corr/abs-2011-08447} studies the semi-random model with a planted clique while the rest of graph is composed of small sized bounded degree graphs, expanders etc. stitched together by a random graph. These works also heavily rely on showing that the vectors corresponding to the planted structure are ``clustered" together and hence using some basic geometric ideas, we can recover a large part of the planted portion. 

\subsection{Preliminaries and Notation} \label{sec:prelims}
Our algorithms are based on the following ``crude SDP''.
	\begin{mybox}
	\begin{SDP}
		\label{sdp:his}
		\[ \max \sum_{\{i_1,i_2,\hdots,i_r\} \in {V \choose r}} \norm{x_{i_1,i_2, \hdots ,i_r}}^2  \]
		\subjectto
		\begin{align}
			\label{eq:sdp1}
			&\norm{x_i}^2 = 1 & \forall i \in V \\
			\label{eq:sdp2}
			&\norm{x_e}^2 = 0 & \forall e \in E \\
			\label{eq:sdp3}
			&\inprod{x_I,x_J} = \norm{x_{I \cup J}}^2  & \forall I,J (\neq \emptyset) \subseteq V, \text{ s.t } \abs{I \cup J} \leq r+1\\
			\label{eq:subset_constraint}
			&\inprod{x_u,x_I} \geq \inprod{x_u,x_J}  & \forall u \in V, \forall I \subseteq J \subseteq V, \abs{J} \leq r+1\\
			\label{eq:union_bound}
			&1- \norm{x_{{u, v_1, \hdots, v_r}}}^2 \leq  \sum_{i \in [r]} \paren{1 - \norm{x_{{u, v_i}}}^2} & \forall \set{u,v_1,\hdots,v_r} \in {V \choose {r+1}} \mper 
		\end{align}
		\end{SDP}
	\end{mybox}

The constraints in \prettyref{sdp:his} are inspired from the Lasserre/SoS hierarchy of constraints. The Lasserre/SoS hierarchy is a strengthened SDP relaxation for nonlinear $0-1$ programs attributed to the works of
Shor \cite{shor1987approach}, Nesterov \cite{MR1748764}, Jean B. Lasserre \cite{MR1814045} and Parrilo \cite{parrilo2003semidefinite}. We refer the reader to the survey by Thomas Rothvo{\ss} \cite{thelasserre} for a detailed discussion.

We also introduce some basic notation that we will be using throughout this paper.
\begin{itemize}
\item
Let $\partial(S)$ or the boundary of $S$ denote ${V \choose r} \setminus \paren{{S \choose r} \cup {{V \setminus S} \choose r}}$. 
\item
Let the optimal solution of the above SDP be denoted by $\set{x^*_{I}}_{I \subset V, 1 \leq \abs{I} \leq r+1}$.
\item
Let $d(v)|_{T}$ be the degree of any vertex $v \in V$, when restricted to only count hyperedges in the set $\set{v} \cup T$.
\item
Throughout the paper, we will assume that $k \leq n/2$, and $r \geq 2$.
\end{itemize}

\subsection{Proof Overview} 

In \cite{McKenzie_2020} they study a crude-SDP with the constraint $\inprod{x_i,x_j}=0, \forall \set{i,j} \in E$. Their crude SDP tries to cluster the vertices together, while the constraint $\inprod{x_i,x_j}=0, \set{i,j} \in E$ tries to ensure that no edges are contained in a cluster. Constraint \prettyref{eq:sdp2} is a natural extension of this to hypergraphs.
We add vectors for all subsets of vertices of size at most $r+1$, and add consistency constraints \prettyref{eq:sdp3} among them, as in the Lasserre/SoS hierarchy. However, we note that \prettyref{sdp:his} is different from a Lasserre/SoS relaxation since there is no natural interpretation of solution to this crude-SDP as a low-degree pseudo-distribution over independent sets in the hypergraph. However, we add the constraints in equation \prettyref{eq:sdp3},\prettyref{eq:subset_constraint} and \prettyref{eq:union_bound} since our intended feasible solution $x'$ constructed as,
	\begin{equation} \label{eq:feasible}
	x_{i_1,i_2,...,i_l}' =  
	\begin{cases} 
	\hat{e} & \text{if } \{i_1,i_2,...,i_l \} \in {S \choose l} \\
	x^*_{i_1,i_2,...,i_l} & \text{if } \{i_1,i_2,...,i_l\} \in {{V \setminus S} \choose l}  \qquad \forall l \in [r+1]\\
	0 &  \text{otherwise}
	\end{cases} 
	\end{equation}
where $\hat{e}$ denote a unit vector orthogonal to $x_{I}^{*}, \ \forall I \subseteq V \setminus S, \ \Abs{I} \leq r$.  
satisfies these constraints (\prettyref{app:feasibility_proof}).
The constraints in equations \prettyref{eq:subset_constraint} and \prettyref{eq:union_bound} are inspired from the locally consistent probability distributions viewpoint of a $r$-level Lasserre/SoS hierarchy \cite{thelasserre}. A $t$-level vector in a Lasserre/SoS hierarchy can be interpreted as the probability of the joint event corresponding to indices of the vector. The constraint \prettyref{eq:subset_constraint} corresponds to the fact that the probability of a sub event can only be larger than the probability of an event and the constraint \prettyref{eq:union_bound} corresponds to a union bound on the complement of joint event (represented by $x_{u,v_1,\hdots,v_r}$) given by sum of complement of pairwise joint events  $1-x_{u,v_i}~\forall i \in [r]$.

In \prettyref{sec:grothendieck}, we prove a lower bound on the contribution of the SDP mass in the optimal solution from the $r$-level vectors of $S$, i.e. $\set{x^*_{I}}_{I \subset S, \abs{I} = r}$ (\prettyref{cor:lower_bound_S}). The high-level idea of our proof is the same as that of \cite{McKenzie_2020}. However, we need some new ideas to extend them to hypergraphs. Using the approach of \cite{McKenzie_2020}, we first lower bound the SDP mass from $S$ and $S \times (V \setminus S)$ (\prettyref{lem:total bound}). Therefore, upper bounding the contribution from $S \times (V \setminus S)$, will give us a lower bound on the contribution from $S$. In \cite{McKenzie_2020}, $S \times (V \setminus S)$ is a random bipartite graph; they use the Grothendieck's inequality and concentration bounds to upper bound the contribution from this part. In our setting, $S \times (V \setminus S)$ is a random hypergraph, and \cite{McKenzie_2020}'s techniques do not seem to be directly applicable here. 
Our main idea is to construct a random bipartite graph $G'= (U_1, U_2, E')$ based on this random bipartite hypergraph as follows (\prettyref{cons:graph}). One side of the graph consists of vertices corresponding to subsets of $S$ of cardinality at most $r-1$, and other side side consists of vertices corresponding to subsets of $V \setminus S$ of cardinality at most $r-1$. We add an edge between two vertices if the union of the sets corresponding to them forms a hyperedge in our hypergraph. By our construction, $\sum_{\set{a,b} \in E'} \inprod{x_a,x_b}$ is equal to the SDP mass from $S \times (V \setminus S)$ in our hypergraph.
Moreover, since $S \times (V \setminus S)$ forms a random bipartite hypergraph, our construction gives us that $G'$ is a random bipartite graph. Therefore, we can now proceed to bounding the contribution from $G'$ using \cite{McKenzie_2020}'s approach (\prettyref{prop:grothendieck_bound}).

Our proof of \prettyref{thm:main2} (in \prettyref{sec:wider_range_alg}) is a generalization of the proof of Theorem 1.1 of \cite{McKenzie_2020} to the case of hypergraphs and our higher order SDP (\prettyref{sdp:his}). \prettyref{cor:lower_bound_S} shows that the $\ell_2^2$ lengths of the $r$-level vectors completely inside $S$ is large. This in turn (by the SDP constraints) implies that there is a vertex $u \in S$ such that most of the $(r-1)$-level vectors in $S$ have a large projection on $x^*_u$ (\prettyref{lem:5.3}). In \cite{McKenzie_2020} they use the SDP constraint $\inprod{x_u,x_v}=0, \forall \set{u,v} \in E$ to show that the set of vectors which have a large projection on $x^{*}_{u}$ is an independent set. Therefore they proceed to bound the parameter regimes to obtain a small value of $p$ and a $(1-\epsilon)k$ lower bound guarantee on the size of this set. However in our setting, for $r \geq 3$, we are unable to guarantee that this set of $(r-1)$ level vectors is an independent set. Therefore, we proceed by using \prettyref{lem:5.3} to show that there exists a vertex $u$ such that a large fraction of the $1$-level vectors in $\set{x^{*}_{v}: v \in S}$ have a large projection ($\geq \mathcal{R'}$) on $x^{*}_{u}$, along the lines of \cite{McKenzie_2020}. Let us consider the set of $1$-level vectors that have a projection $\geq \mathcal{R'}$ on $x^{*}_{u}$ (\prettyref{def:ball}). Showing that the $r$-level vectors consisting of vertices from this set have non-zero norm will suffice to guarantee that there are no hyperedges in this set. We use our ``union-bound'' SDP constraint \prettyref{eq:union_bound} in our crude-SDP to establish this (\prettyref{lem:5.6}). Choosing  $\mathcal{R'}$ to be large enough $\paren{\mathcal{R'} = 1 - \ffrac{1}{2r}}$ and using the SDP constraint \prettyref{eq:union_bound}, we establish a non-zero lower bound on $\inprod{x^*_u,x^*_{v_1,v_2,\hdots,v_r}}$ for every $r$-tuple $(v_1,\hdots,v_r)$ consisting of vertices from the set. Now using SDP constraints \prettyref{eq:sdp2} and \prettyref{eq:subset_constraint}, we can establish that the vertices inside the set do not form a hyperedge. For our choice of parameters in this theorem, the set of vertices corresponding to this set will contain at least $(1-\e)$ fraction of the vertices in $S$.

Our proof of \prettyref{thm:main} (in \prettyref{sec:full_recovery}) is a generalization of the proof of Theorem 1.2 of \cite{McKenzie_2020} to the case of hypergraphs and our higher order SDP (\prettyref{sdp:his}). In \prettyref{lem:5.3} we show that there exists a vertex $u \in S$ such that most of the $(r-1)$-level vectors in $S$ have a large projection on $x^*_u$. Let us consider the set of $(r-1)$-level vectors which have a large projection ($\geq \mathcal{R}$) on $x^{*}_{u}$ (\prettyref{def:ball}). The choice of $p$ ensures that each vertex in $V \setminus S$ forms a hyperedge with at least one of the tuples corresponding to $(r-1)$ level vectors in the set w.h.p. (\prettyref{lem:4.6}).
Moreover, the choice of $\mathcal{R}$ ensures that the set can not contain two orthogonal vectors (\prettyref{lem:4.4}). Therefore, this ensures that the tuples in set contains vertices only from $S$  (\prettyref{lem:4.7}). Therefore, the union of the sets of vertices contained in the $(r-1)$-tuples corresponding to such $(r-1)$-level vectors would be a subset of $S$. A greedy algorithm can be used to recover the remaining vertices of $S$. Since we don't know this special vertex $u$, we perform this procedure on each vertex and return the set of independent sets obtained; one of these independent sets would be the planted one w.h.p. The whole procedure is presented in Algorithm \ref{alg:one}.
The range of $p$ in this theorem is however smaller than the range of $p$ for which \prettyref{thm:main2} is guaranteed to hold.

\section{Bounding the contribution from the random hypergraph}
\label{sec:grothendieck}

In this section, we bound the contribution of the SDP (\prettyref{sdp:his}) mass from the random portion of the hypergraph. As a result, we find a lower bound on the contribution of the vectors from our planted independent set $S$. The two key technical results (\prettyref{prop:grothendieck_bound} and \prettyref{cor:lower_bound_S}) which we prove in this section which generalize (\cite{McKenzie_2020}, Lemma 2.1) to $r$-uniform hypergraphs are the following.

\begin{proposition}
	\label{prop:grothendieck_bound}
	For $k \geq \dfrac{r2^{2r+2}e^{r}}{3p}$,
	\[
	\sum_{\set{i_1, i_2, \hdots, i_r} \in \partial(S)} \norm{x_{i_1, i_2, \hdots, i_r}^{*}}^2 \leq \paren{\dfrac{2^{3r-2}e^{3r/2-2}}{\sqrt{3}r^{r-5/2}}}\paren{\sqrt{\frac{k}{p}}}n^{r-1}\mper
	\]
	with high probability (over the randomness of the input).
\end{proposition}

\begin{corollary}
	\label{cor:lower_bound_S}
	For $k \geq \dfrac{r2^{2r+2}e^{r}}{3p}$,
	\[
	\sum_{\set{i_1, i_2, \hdots, i_r} \in {S \choose r}} \norm{x_{i_1, i_2, \hdots, i_r}^{*}}^2 \geq 
	{k \choose r} - \paren{\dfrac{2^{3r-2}e^{3r/2-2}}{\sqrt{3}r^{r-5/2}}}\paren{\sqrt{\frac{k}{p}}}n^{r-1}\mper
	\]
	with high probability (over the randomness of the input).
\end{corollary}

The main lemma which connects the above two results is as follows. 

\begin{lemma}\label{lem:total bound}
	\[ \sum_{\{i_1i_2 \hdots i_r\} \in {S \choose r}} \norm{x_{i_1,i_2, \hdots ,i_r}^*}^2 +  \sum_{\{i_1,i_2,\hdots,i_r\} \in \partial(S)} \norm{x_{i_1,i_2, \hdots, i_r}^*}^2 \geq {k \choose r} \mper \]
\end{lemma}
\begin{proof}
We start by inspecting our intended feasible solution $x'$ as defined in equation \prettyref{eq:feasible}.
A straightforward calculation shows that $x'$ is indeed a feasible solution of the SDP\footnote{Our definition of $x'$ depends on existence of $x^*$, we can show $x^*$ does exist by exhibiting a solution $x''$ that satisfies all the constraints as $x''=e_i, \forall i \in V$, where $\set{e_i}_{i=1}^{n}$ are orthonormal and $x''=0$ otherwise, and gives a SDP value of 0.}, thus we defer these details to \prettyref{app:feasibility_proof}. Note that by splitting the sum into three disjoint parts, we have
	\begin{align}
	\nonumber
	\sum_{\{i_1,i_2,\hdots,i_r\} \in {V \choose r}} \norm{x_{i_1,i_2, \hdots, i_r}^*}^2 &= \sum_{\{i_1,i_2, \hdots, i_r\} \in {S \choose r}} \norm{x_{i_1,i_2, \hdots, i_r}^*}^2 + \sum_{\{i_1,i_2,\hdots,i_r\} \in \partial(S)} \norm{x_{i_1,i_2, \hdots, i_r}^*}^2 \\
	\label{eq:x_star}
	&\qquad\qquad+ \sum_{\{i_1,i_2, \hdots, i_r\} \in {V \setminus S \choose r}} \norm{x_{i_1,i_2, \hdots, i_r}^*}^2
	\end{align}
	and similarly,
	\begin{align}
	\nonumber
	\sum_{\{i_1,i_2,\hdots,i_r\} \in {V \choose r}} \norm{x_{i_1,i_2, \hdots, i_r}'}^2 &= \sum_{\{i_1,i_2, \hdots, i_r\} \in {S \choose r}} \norm{x_{i_1,i_2, \hdots, i_r}'}^2 +  \sum_{\{i_1,i_2,\hdots,i_r\} \in \partial(S)} \norm{x_{i_1,i_2, \hdots, i_r}'}^2\\ 
	\nonumber
	&\qquad\qquad+ \sum_{\{i_1,i_2, \hdots, i_r\} \in {V \setminus S \choose r}} \norm{x_{i_1,i_2, \hdots, i_r}'}^2\\
	\label{eq:x_prime}
	&= {k \choose r} + \sum_{\{i_1,i_2, \hdots, i_r\} \in {V \setminus S \choose r}} \norm{x_{i_1,i_2, \hdots, i_r}^*}^2\qquad(\text{from eqn \prettyref{eq:feasible}})\mper
	\end{align}
	Since $x^{*}$ is optimal we have that,
	\begin{align*}
	\sum_{\{i_1,i_2,\hdots,i_r\} \in {V \choose r}} \norm{x_{i_1,i_2, \hdots, i_r}^*}^2 &\geq \sum_{\{i_1,i_2,\hdots,i_r\} \in {V \choose r}} \norm{x_{i_1,i_2, \hdots, i_r}'}^2\\
	\implies\sum_{\{i_1,i_2, \hdots, i_r\} \in {S \choose r}} \norm{x_{i_1,i_2, \hdots, i_r}^*}^2 +  \sum_{\{i_1,i_2,\hdots,i_r\} \in \partial(S)} \norm{x_{i_1,i_2, \hdots, i_r}^*}^2 &\geq {k \choose r}\quad(\text{from eqns \prettyref{eq:x_star} \& \prettyref{eq:x_prime}})\mper
	\end{align*}
\end{proof}

Note that the above lemma which is similar to (\cite{McKenzie_2020}, Lemma 2.2) helps us remove the dependence of the contribution of the vectors from $V \setminus S$, is the key lemma which allows us to work with an arbitrary subhypergraph $H\brac{V \setminus S}$. Also, it makes our arguments invariant to any extra hyperedges added by an adversary.

Next, we proceed to prove \prettyref{prop:grothendieck_bound}. We begin by constructing a bipartite graph to simplify our calculations, as follows.

\begin{construction}
\label{cons:graph}
We construct a bipartite graph $G' \defeq (U_1,U_2, E')$ from the given input hypergraph $H$ as follows.
\[ \text{Here }
	U_1 \defeq (S) \cup {S \choose 2} \cup \hdots \cup {S \choose r-1} \text{ and }\\
	U_2 \defeq \paren{V \setminus S} \cup {{V \setminus S} \choose 2} \cup \hdots \cup {{V \setminus S} \choose r-1}\mper
\]
Now for each hyperedge $e$ in our original hypergraph $H$ (before the action of the monotone adversary on $H\brac{S, V \setminus S}$) such that $e \in E \cap \partial(S)$, let $I_e \defeq e \cap S$ and $J_e \defeq e \cap (V \setminus S)$. We add an edge in the graph $G'$ between the vertices $I_e \in U_1$ and $J_e \in U_2$. It is easy to see that there is a bijection between the random part of the hypergraph and $G'$.
\end{construction}
Let $A$ denote the adjacency matrix of $G'$ (of dimension $\abs{U_1} + \abs{U_2}$) and let $m'$ denote the maximum number of number of edges in the random hypergraph.

In the next few lemmas, we setup some groundwork to use this construction in establishing our claims. We prove the following bounds on $\abs{U_1}$, $\abs{U_2}$ and $m'$. The proof uses some standard results on binomial coefficients. For completeness, we state them in \prettyref{fact:sterling}.

\begin{fact}\label{fact:set_bounds}
	For all $k \leq n/2, r \geq 2$ we have,
	\begin{multicols}{2}
	\begin{enumerate}
		\item $1+\abs{U_1} \leq r\paren{\dfrac{2ek}{r}}^{r-1}\mper$
		\item $1+\abs{U_2} \leq r\paren{\dfrac{2en}{r}}^{r-1}\mper$
		\item $m' \leq \dfrac{(4e)^{r-2}kn^{r-1}}{r^{r-2}}\mper$
		\item $m' \geq k\paren{\dfrac{n}{2r}}^{r-1}\mper$
	\end{enumerate}
	\end{multicols}
\end{fact}
\begin{proof}
	By using \prettyref{fact:sterling},
	\begin{enumerate}
		\item
		\begin{align*}
			1+\abs{U_1} = {k \choose 0} + {k \choose 1} + {k \choose 2} + \hdots + {k \choose {r-1}} \leq r {k \choose r-1} \leq r\paren{\dfrac{ek}{r-1}}^{r-1} \leq r\paren{\dfrac{2ek}{r}}^{r-1}\mper
		\end{align*}
		where we use the fact that $r- 1 \geq \ffrac{r}{2} \iff r \geq 2$ in the last step.
		\item Similarly,
		\begin{align*} 
		1+\abs{U_2} &= {n-k \choose 0} + {n-k \choose 1} + {n-k \choose 2} + \hdots + {n-k \choose {r-1}}\\
		&\leq r {n-k \choose r-1}
		\leq r \left({\frac{e(n-k)}{r-1}}\right)^{r-1} \leq r\paren{\dfrac{2en}{r}}^{r-1}\mper
		\end{align*}
		\item 
		Since only possible edges are between subsets of size $i$ in $U_1$ and size $r-i$ in $U_2$, we can write $m'$ as
		\begin{align*}
		m' = \sum_{i=1}^{r-1} {k \choose i}{{n-k} \choose r-i} \leq {k \choose 1}{{n-k} \choose 1}{n-2 \choose r-2} \leq k(n-k)\paren{\dfrac{e(n-2)}{r-2}}^{r-2} \leq \dfrac{(4e)^{r-2}kn^{r-1}}{r^{r-2}} \mper
		\end{align*}
		The first inequality follows from the fact that every possible hyperedge in $\partial(S)$ can be upper bounded by picking $r$-tuples where at least one vertex is chosen from $S$ and $V \setminus S$ each and rest $r-2$ vertices are chosen arbitrarily. The last inequality follows form the fact that $r \geq 3$ and $n \geq n - 2~\&~n - k$. The inequality is not applicable for $r=2$ but we can show the bound on $m'$ still holds by computing $m'$ exactly.
		\item \begin{align*}
		m' = \sum_{i=1}^{r-1} {k \choose i}{{n - k} \choose r-i} \geq {k \choose 1}{{n-k} \choose {r-1}} \geq k\paren{\dfrac{n-k}{r-1}}^{r-1} \geq k\paren{\dfrac{n}{2r}}^{r-1}\mper 
		\end{align*}
		where we use the fact that $k \leq n/2$ and $r \geq r - 1$ in the last step.
	\end{enumerate}
\end{proof}

\begin{definition} \label{def:hyper2graph}
	We define a centered matrix $B \in \R^{\paren{\abs{U_1} + \abs{U_2}} \times \paren{\abs{U_1} + \abs{U_2}}}$,
	\[
	B_{I,J} \defeq
	\begin{cases}
	p-A_{I,J} &\quad \forall i \in [r-1], I \in {S \choose i}, J \in {{V \setminus S} \choose r-i}; \forall j \in [r-1], I \in {{V \setminus S} \choose j}, J \in {{S} \choose r-j}\\
	0 &\quad \text{otherwise}\mper
	\end{cases}
	\]
\end{definition}
\noindent
where $A$ denotes the adjacency matrix of $G'$ in  \prettyref{cons:graph}. Note that by construction, $\mathbbm{E}\brac{B} = 0$. We rewrite the contribution of the random hypergraph towards the SDP mass in terms of the matrix $B$ using the next lemma.

\begin{lemma}
	\label{lem:cross_terms}
	\begin{align*}
	\sum_{\set{i_1, i_2, \hdots, i_r} \in \partial(S)} \norm{x_{i_1, i_2, \hdots, i_r}^{*}}^2
	&=  \dfrac{1}{2p} \paren{\sum_{u_1,u_2 \in U_1 \cup U_2}{B_{u_1,u_2}\inprod{x_{u_1}^{*},x_{u_2}^{*}}}}\mper
	\end{align*}
\end{lemma}
\begin{proof}
	\begin{align*}
	\sum_{\set{i_1, i_2, \hdots, i_r} \in \partial(S)} \norm{x_{i_1, i_2, \hdots, i_r}^{*}}^2 
	&= \sum_{\set{i_1, i_2, \hdots, i_r} \in \partial(S)} \norm{x_{i_1, i_2, \hdots, i_r}^{*}}^2 - \dfrac{1}{p} \sum_{e \in \partial(S) \cap E} \norm{x^{*}_{e}}^2 \qquad\paren{\because \norm{x^{*}_e}^2=0,  \forall e \in E}\\
	&= \dfrac{1}{p}\paren{\sum_{\set{i_1, i_2, \hdots, i_r} \in \partial(S)} \paren{p-\one_{\set{i_1, i_2, \hdots, i_r} \in E}} \norm{x_{i_1, i_2, \hdots, i_r}^{*}}^2} \qquad\paren{\text{Combining the sum}}\\
	&= \dfrac{1}{p}\paren{\sum_{i=1}^{r-1}\sum_{u_1 \in {S \choose i}, u_2 \in {{V \setminus S} \choose {r-i}}} \paren{p-A_{u_1, u_2}} \inprod{x_{u_1}^{*}, x_{u_2}^{*}}} \qquad\paren{\text{by SDP constraint } \prettyref{eq:sdp3}}\\
	&= \dfrac{1}{p}\paren{\sum_{i=1}^{r-1}\sum_{u_1 \in {S \choose i}, u_2 \in {{V \setminus S} \choose {r-i}}} B_{u_1, u_2} \inprod{x_{u_1}^{*}, x_{u_2}^{*}}} \qquad\paren{\text{by \prettyref{def:hyper2graph}} }\\
	&= \dfrac{1}{p} \paren{\sum_{u_1 \in U_1,u_2 \in U_2}{B_{u_1,u_2}\inprod{x_{u_1}^{*},x_{u_2}^{*}}}} = \dfrac{1}{2p} \paren{\sum_{u_1,u_2 \in U_1 \cup U_2}{B_{u_1,u_2}\inprod{x_{u_1}^{*},x_{u_2}^{*}}}}
	\mper
	\end{align*}
\end{proof}

It is important to note that the above lemma rewrites the mass of the SDP by vectors in the boundary of $S$ (the random part) using the matrix $B$. The entries of $B$ only depend on the initial set of random edges, thus any extra edges added by a monotone adversary can be ignored w.l.o.g.

We are now ready to prove \prettyref{prop:grothendieck_bound}. The proof uses some commonly used concentration inequalitites. The exact variants of these are stated in \prettyref{fact:bern} and \prettyref{fact:groth}.
\begin{proof}[Proof of \prettyref{prop:grothendieck_bound}]
	We start by bounding the term
	$
		\sum_{u_1, u_2 \in U_1 \cup U_2} B_{u_1, u_2} \inprod{x_{u_1}^{*}, x_{u_2}^{*}}\mper
	$	
	Since ${\norm{x_I^*}}^2 \leq 1, \; \forall I \subset V,\;\abs I \leq r$ we can use Grothendieck's inequality (\prettyref{fact:groth}) to bound it. We restate it here.
	\begin{align*}
		\max_{\substack{x_{u_1}, x_{u_2}:u_1,u_2 \in U_1 \cup U_2\\ \norm{x_{u_1}}, \norm{x_{u_2}} \leq 1}} \abs{\sum_{u_1, u_2 \in U_1 \cup U_2} B_{u_1, u_2} \inprod{x_{u_1}, x_{u_2}}} &\leq
		2 \max_{\substack{y_{u_1}, y_{u_2}:u_1,u_2 \in U_1 \cup U_2\\ {y_{u_1}}, {y_{u_2}} \in \{\pm{1}\}}} \abs{\sum_{u_1, u_2 \in U_1 \cup U_2} B_{u_1, u_2} y_{u_1} y_{u_2}}\\
		&\leq 4 \max_{\substack{y_{u_1}, y_{u_2}:u_1,u_2 \in U_1 \cup U_2\\ {y_{u_1}}, {y_{u_2}} \in \{\pm{1}\}}} \abs{\sum_{u_1 \in U_1,u_2 \in U_2} B_{u_1, u_2} y_{u_1} y_{u_2}}\mper \numberthis \label{eq:discretize_term}
	\end{align*}
	
	For a fixed set of variables, $y_{u_1}, y_{u_2}$ and a parameter $\delta \in (0, 1]$ to be fixed later, we use Bernstein's inequality (\prettyref{fact:bern}) on $m'$ independent random variables $B_{u_1,u_2}y_{u_1}y_{u_2}$ (each of which has mean $0$, is bounded by $1$, and has a variance atmost $p$). Then for $t =\delta p m'$ and for all $\delta \in (0,1]$, and using an upper bound on $m'$ from \prettyref{fact:set_bounds} we have that,
	\begin{align*} 
	\ProbOp\left[\abs{
			\sum_{u_1 \in U_1,u_2 \in U_2} B_{u_1, u_2} y_{u_1} y_{u_2}
		} > \delta p \dfrac{(4e)^{r-2}kn^{r-1}}{r^{r-2}}
	\right]
	&\leq \ProbOp\left[\abs{
		\sum_{u_1 \in U_1,u_2 \in U_2} B_{u_1, u_2} y_{u_1} y_{u_2}
	} > \delta p m'\right]\\
	&\leq 2 \exp\paren{-\dfrac{{\delta}^2 p^2 m'}{2p + 2\delta p/3}}\\
	&\leq 2 \exp\paren{-\dfrac{3{\delta}^2 p m'}{8}} \qquad\qquad\paren{\because \delta \leq 1}\mper
\end{align*}
	By a union bound over all possible values of $y_{u_1},y_{u_2}$,
		\begin{align*}
	&\ProbOp\left[
	\max_{\substack{y_{u_1}, y_{u_2}:u_1,u_2 \in U_1 \cup U_2\\ {y_{u_1}}, {y_{u_2}} \in \{\pm{1}\}}}{
		\abs{\sum_{u_1 \in U_1, u_2 \in U_2} B_{u_1,u_2} y_{u_1} y_{u_2}} > \delta p \dfrac{(4e)^{r-2}kn^{r-1}}{r^{r-2}}
	}\right]\\ 
	&\qquad\qquad\qquad\qquad\leq 2^{1+2\paren{\abs{U_1}+\abs{U_2}}} \exp\paren{-\dfrac{3{\delta}^2 pm'}{8}}\\
	&\qquad\qquad\qquad\qquad\leq \exp\paren{2\paren{1+\abs{U_1}}+2\paren{1+\abs{U_2}}-\dfrac{3{\delta}^2 pm'}{8}}\qquad \\
	&\qquad\qquad\qquad\qquad\leq \exp\paren{2r\paren{\dfrac{2ek}{r}}^{r-1} + 2r\paren{\dfrac{2en}{r}}^{r-1} - \dfrac{3\delta^2pk}{8}\paren{\dfrac{n}{2r}}^{r-1}}\qquad \paren{\text{from \prettyref{fact:set_bounds}}}\\
	&\qquad\qquad\qquad\qquad\leq \exp\paren{4r\paren{\dfrac{2en}{r}}^{r-1} - \dfrac{3\delta^2pk}{8}\paren{\dfrac{n}{2r}}^{r-1}}\qquad \paren{\text{since } k < n}\\
	&\qquad\qquad\qquad\qquad= \exp\paren{-\paren{\dfrac{3\delta^2pk}{r2^{r+2}}-4(2e)^{r-1}}\dfrac{n^{r-1}}{r^{r-2}}}\\
	&\qquad\qquad\qquad\qquad= \exp\paren{-\dfrac{n^{r-1} (2e)^{r-1} 2(e-2)}{r^{r-2}}}\qquad \paren{\text{we set }\delta^2 = \frac{r2^{2r+2}e^{r}}{3pk}}\\
	&\qquad\qquad\qquad\qquad\leq \exp\paren{-n} \qquad \paren{\because \dfrac{n^{r-1}}{r^{r-2}} \geq n\paren{\dfrac{n}{r}}^{r-2}\text{ where } n/r \geq 1 \text{ and } 2(e-2)(2e)^{r-1} \geq 1} \mper
	\end{align*}

Note that this holds when $\delta^2 \leq 1 \iff \dfrac{r2^{2r+2}e^{r}}{3pk} \leq 1 \iff k \geq \dfrac{r2^{2r+2}e^{r}}{3p}$.
Therefore using the union bound above and equation \prettyref{eq:discretize_term}, we have that with high probability (for enough large $n$ and when $\delta^2 \leq 1$),
\[\sum_{u_1, u_2 \in U_1 \cup U_2} B_{u_1, u_2} \inprod{x_{u_1}^{*}, x_{u_2}^{*}} \leq 4\delta p \dfrac{(4e)^{r-2}kn^{r-1}}{r^{r-2}}\mper\]
Using \prettyref{lem:cross_terms} we have that,
\[
\sum_{\set{i_1, i_2, \hdots, i_r} \in \partial(S)} \norm{x_{i_1, i_2, \hdots, i_r}^{*}}^2 \leq 2 \delta \dfrac{(4e)^{r-2}kn^{r-1}}{r^{r-2}} 
=\paren{\dfrac{2^{3r-2}e^{3r/2-2}}{\sqrt{3}r^{r-5/2}}}\paren{\sqrt{\frac{k}{p}}}n^{r-1}
\]
where we substitute the value of $\delta$ to complete the proof.
\end{proof}

We define the following function for notational convenience.

\begin{definition}\label{def:const_f}
	Let $f(r) \defeq \dfrac{r^{5/2}2^{3r-2}e^{3r/2-2}}{\sqrt{3}} \mper$
\end{definition}

\begin{proof}[Proof of \prettyref{cor:lower_bound_S}]
	The proof follows almost immediately from \prettyref{prop:grothendieck_bound} and \prettyref{lem:total bound},
	\begin{align*}
		\sum_{\{i_1i_2 \hdots i_r\} \in {S \choose r}} \norm{x_{i_1i_2 \hdots i_r}^*}^2
		\geq {k \choose r}- \sum_{\{i_1,i_2,\hdots,i_r\} \in \partial(S)} \norm{x_{i_1, i_2, \hdots, i_r}^{*}}^2 &\geq {k \choose r} - \paren{\dfrac{2^{3r-2}e^{3r/2-2}}{\sqrt{3}r^{r-5/2}}}\paren{\sqrt{\frac{k}{p}}}n^{r-1}\\
		&= {k \choose r} - \dfrac{f(r)n^{r-1}\sqrt{k}}{r^r\sqrt{p}}
		\mper
	\end{align*}
\end{proof}

\section{Algorithm for computing a large independent set}
\label{sec:wider_range_alg}
In this section, we will prove a formal version of \prettyref{thm:main2} which is a generalization of Theorem 1.1 of \cite{McKenzie_2020} to $r$-uniform hypergraphs (\prettyref{lem:5.3}, \prettyref{lem:5.5} and proof of \prettyref{thm:main2}). We will crucially use the lower bound on the SDP mass from the vectors in S, i.e.,  \prettyref{cor:lower_bound_S}. 

As a first step towards this, in \prettyref{lem:5.3}, we show that there exists a vertex $u \in S$ for which the 1 level vectors $x^*_v$ (corresponding to vertices in $S$) in the optimal solution have a large projection on $x^*_u$.

\begin{lemma}
	\label{lem:5.3}

	For $k \geq \dfrac{r2^{2r+2}e^{r}}{3p}$, there exists a vertex $u \in S$ such that, with high probability (over the randomness of the input).

	\[ \mathbbm{E}_{v \in S \setminus \set{u}}\inprod{x^{*}_u,x^{*}_v} \geq \mathbbm{E}_{\set{i_1,i_2,\hdots,i_{r-1}} \sim {S \setminus \set{u} \choose r-1}} {\inprod{x^{*}_u,x^{*}_{i_1,i_2,\hdots,i_{r-1}}}} \geq 1- \dfrac{f(r)n^{r-1}}{k^{r-0.5}\sqrt{p}} \mper \]
\end{lemma}
\begin{proof}
From \prettyref{cor:lower_bound_S} we have that for  $k \geq \dfrac{r2^{2r+3}e^{r}}{3p}$,
	\[ \sum_{\{i_1,i_2,\hdots,i_r\} \in {S \choose r}}{\norm{x^{*}_{i_1,i_2,\hdots,i_r}}}^2 \geq {k \choose r} - \dfrac{f(r)n^{r-1}\sqrt{k}}{r^r\sqrt{p}} \mper\]
From the SDP constraint \prettyref{eq:sdp3}, we split the above sum as follows,
	\begin{equation}\sum_{u \in S, \set{i_1,i_2,\hdots,i_{r-1}} \in {S \setminus \set{u} \choose r-1}}{\inprod{x^{*}_u,x^{*}_{i_1,i_2,\hdots,i_{r-1}}}} \geq r\paren{{k \choose r} - \dfrac{f(r)n^{r-1}\sqrt{k}}{r^r\sqrt{p}}} \mper \label{eq:summation_inprod_bound} \end{equation}
Therefore there exists a vertex $u \in S$ such that,
	\[ \sum_{\set{i_1,i_2,\hdots,i_{r-1}} \in {S \setminus \set{u} \choose r-1}}{\inprod{x^{*}_u,x^{*}_{i_1,i_2,\hdots,i_{r-1}}}} \geq \dfrac{r}{k}\paren{{k \choose r} - \dfrac{f(r)n^{r-1}\sqrt{k}}{r^r\sqrt{p}}} \mper\]
Since number of terms in expression in the above sum  is ${k-1} \choose {r-1}$. We rewrite the above expression as an expectation over the uniform distribution on such tuples as,
\begin{align*}
	\mathbbm{E}_{\set{i_1,i_2,\hdots,i_{r-1}} \sim {S \setminus \set{u} \choose r-1}} {\inprod{x^{*}_u,x^{*}_{i_1,i_2,\hdots,i_{r-1}}}} &\geq \dfrac{r}{k{{k-1} \choose {r-1}}}\paren{{k \choose r} - \dfrac{f(r)n^{r-1}\sqrt{k}}{r^r\sqrt{p}}}
	= 1 - \dfrac{rf(r)n^{r-1}\sqrt{k}}{k{{k-1} \choose {r-1}}r^r\sqrt{p}}\\
	&\geq 1- \dfrac{rf(r)n^{r-1}\sqrt{k}}{k\paren{\dfrac{k-1}{r-1}}^{r-1}r^r\sqrt{p}}
	\geq 1- \dfrac{rf(r)n^{r-1}\sqrt{k}}{k\paren{\dfrac{k}{r}}^{r-1}r^r\sqrt{p}}\\
	&= 1- \dfrac{f(r)n^{r-1}}{k^{r-0.5}\sqrt{p}}\mper
\end{align*}
where we used \prettyref{fact:sterling} and the fact that, $\flatfrac{(k-1)}{(r-1)} \geq \flatfrac{k}{r} \iff k \geq r$.

Using our SDP constraint \prettyref{eq:subset_constraint} we can rewrite the summation in \prettyref{eq:summation_inprod_bound} as,
\begin{align} 
\nonumber
\sum_{u\in S,\set{i_1,i_2,\hdots,i_{r-1}} \in {S \setminus \set{u} \choose r-1}} {\inprod{x^{*}_u,x^{*}_{i_1,i_2,\hdots,i_{r-1}}}} 
& \leq \dfrac{1}{(r-1)}\sum_{u\in S,\set{i_1,i_2,\hdots,i_{r-1}} \in {S \setminus \set{u} \choose r-1}} \sum_{l=1}^{r-1}{\inprod{x^*_u,x^*_{i_l}}}\\
\label{eq:relate_vectors}
&=\dfrac{{k-2 \choose r-2}}{(r-1)}
\sum_{u \in S, v \in S \setminus \set{u}}{\inprod{x^{*}_u,x^{*}_v}} 
\end{align}
\noindent
where the equality above can be argued by fixing a vertex $u \in S,v \in S \setminus \set{u}$ and observing that there are ${k-2 \choose r-2}$ terms in the double summation  containing such $(u,v)$. We divide the equation \prettyref{eq:relate_vectors} by $k{k-1 \choose r-1}$ (the number of terms in the summation on the left side) to rewrite the inequality in form of expectation as,
\begin{align*}
\mathbb{E}_{\set{i_1,i_2,\hdots,i_{r-1}} \sim {S \setminus \set{u} \choose r-1}} {\inprod{x^{*}_u,x^{*}_{i_1,i_2,\hdots,i_{r-1}}}} 
&\leq \dfrac{{k-2 \choose r-2}}{(r-1)k{k-1 \choose r-1}}
\sum_{u \in S, v \in S \setminus \set{u}}{\inprod{x^{*}_u,x^{*}_v}}\\
&= \dfrac{1}{k(k-1)}\sum_{u \in S, v \in S \setminus \set{u}}{\inprod{x^{*}_u,x^{*}_v}} =  
\mathbb{E}_{v \in S \setminus \set{u}}\inprod{x^{*}_u,x^{*}_v} 
\end{align*}
\noindent
where we have used the fact that ${k-1 \choose r-1}= \frac{k-1}{r-1} {k-2 \choose r-2}$.
It then follows that there exists a vertex $u \in S$ such that
\begin{align*}
\mathbb{E}_{v \in S \setminus \set{u}}\inprod{x^{*}_u,x^{*}_v} 
\geq 1 - \dfrac{f(r)n^{r-1}}{\sqrt{p}k^{r-0.5}}
\end{align*}
\end{proof}

\prettyref{lem:5.3} shows that a large fraction of the 1-level vectors in $S$ have a large projection on $x^*_u$. We start with the following definition,

\begin{definition}
	\label{def:ball}
	We denote the set of all $l$-tuples containing vertices from a set $T \subseteq V$ (where $l \leq \abs{T}$) whose corresponding vectors have a projection at least $\mathcal{R}$ with the vector $x_u^*$ by
	\[
	\mathcal{B}_{u}(l, \mathcal{R}, T) \defeq \set{\set{v_1, v_2, \hdots, v_{l}} : \set{v_1, v_2, \hdots, v_{l}} \in {T \choose l} \text{ and } \inprod{x_u^*,x_{{v_1, v_2, \hdots, v_{l}}}^*} \geq \mathcal{R}}\mper
	\]
\end{definition}

Note that the typical values of $l$ of interest will be $1$ in \prettyref{thm:main2} and $r-1$ in \prettyref{thm:main}. 

\begin{lemma}\label{lem:5.5}
	For $k \geq \dfrac{r2^{2r+2}e^{r}}{3p}$, there exists a vertex $u \in S$ such that
	\[ \abs{\mathcal{B}_u\paren{1,1 - \dfrac{1}{2r}, S}} \geq 
	(k-1)\paren{1 - \dfrac{2rf(r)n^{r-1}}{\sqrt{p}k^{r-0.5}}} \]
	with high probability (over the randomness of the input).
\end{lemma}
\begin{proof}\label{app:5.5-proof}
    We note that ${1-\inprod{x^{*}_u,x^{*}_v}} \geq 0$ and for $\mathcal{R} \in \paren{0,1}$ and for  $k \geq \dfrac{r2^{2r+3}e^{r}}{3p}$, by applying Markov's inequality on $\paren{1 - \inprod{x^{*}_u,x^{*}_v}}$, where $u$ is the vertex guaranteed in \prettyref{lem:5.3} and $v \in V \setminus S$ we have that,
\begin{align*} 
\mathbb{P}_{v \in S \setminus \set{u}}\left[{1-\inprod{x^{*}_u,x^{*}_v}} > {1-\mathcal{R}}\right] < \dfrac{\dfrac{f(r)n^{r-1}}{\sqrt{p}k^{r-0.5}}}{1-\mathcal{R}} \mper && \text{(using \prettyref{lem:5.3})}
\end{align*}
 We can rewrite the above expression as the fraction of vertices which satisfy  $\paren{\inprod{x^{*}_u,x^{*}_v} < \mathcal{R}}$, since the underlying distribution is the uniform distribution over all such $v$ and  by setting $\mathcal{R} = 1-\ffrac{1}{2r}$,
\[ \abs{v \in S \setminus \set{u} : \inprod{x^{*}_u,x^{*}_v} < 1-\dfrac{1}{2r}} < (k-1)\paren{\dfrac{2rf(r)n^{r-1}}{\sqrt{p}k^{r-0.5}}} \mper \]
\[ \therefore \abs{\mathcal{B}_u\paren{1,1-\dfrac{1}{2r}, S}} = \abs{v \in S \setminus \set{u} : \inprod{x^{*}_u,x^{*}_v} \geq 1-\dfrac{1}{2r}} \geq (k-1)\paren{1 - \dfrac{2rf(r)n^{r-1}}{\sqrt{p}k^{r-0.5}}} \mper \]
\end{proof}

In \cite{McKenzie_2020} they use the SDP constraint $\inprod{x_u,x_v}=0, \forall \set{u,v} \in E$ to show that the set of vectors which have a large projection on $x^{*}_{u}$ is an independent set. Therefore they directly analyze the bound on the size of the set to obtain an independent set, in a range of $p$ such that it covers atleast $\paren{1-\e}$ fraction of vertices in S. However in our setting, we are unable to guarantee directly that this set of vectors is an independent set. We crucially use the Lasserre/SoS like SDP constraints \prettyref{eq:sdp3} and \prettyref{eq:union_bound} and an appropriately large value of $\mathcal{R}$ ($\mathcal{R} \geq 1-\frac{1}{2r}$) to show that the set guaranteed in \prettyref{lem:5.5} is an independent set.

\begin{lemma}\label{lem:5.6}
For $k \geq \dfrac{r2^{2r+2}e^{r}}{3p}$, there exists a vertex $u \in S$ such that $\mathcal{B}_u\paren{1,1 - \dfrac{1}{2r},V}$ is an independent set with high probability (over the randomness of the input).
\end{lemma}
\begin{proof}
We consider the SDP constraint \prettyref{eq:union_bound} and apply it to our optimal solution $x^*$ . By using consistency constraints $\paren{\inprod{x_I,x_J} = \inprod{x_{I'}.x_{J'}}, \forall I \cup J = I' \cup J'}$ (equation \prettyref{eq:sdp3}) we can rewrite the constraint in \prettyref{eq:union_bound} as,
\begin{equation}
	\label{eq:union_bound_apply}
	1 -\norm{x^{*}_{u,i_1,\hdots,i_r}}^2  \leq \sum_{l \in [r]} \paren{{1 - \inprod{x^{*}_u,x^{*}_{i_l}}}}\mper
\end{equation}
For  $k \geq \dfrac{r2^{2r+3}e^{r}}{3p}$, if we pick any set of $r$ vertices $\set{i_1,\hdots,i_r} \in {V \choose {r}} $ in $\mathcal{B}_u\paren{1,1-\dfrac{1}{2r},V}$ (where $u$ is the vertex guaranteed in \prettyref{lem:5.5}) we know that $\inprod{x^{*}_u,x^{*}_{i_{l}}} \geq 1-\dfrac{1}{2r}, \forall l \in [r]$. By using equation \prettyref{eq:union_bound_apply} we have that,
\begin{equation}
	\label{eq:hyperedge_term_bound}
	\norm{x^{*}_{u,i_1,\hdots,i_r}}^2 \geq 1 -  \sum_{l \in [r]}{\paren{{1 - \inprod{x^{*}_u,x^{*}_{i_l}}}}} \geq 1 - \sum_{l \in [r]}{\dfrac{1}{2r}} \geq \dfrac{1}{2} > 0 \mper
\end{equation}
Now we examine the term $\norm{x_{i_1,i_2, \hdots,i_r}^*}^2$ for these $\set{i_!,\hdots,i_r}$ and we have that,
\begin{align*}
    \norm{x_{i_1,i_2, \hdots,i_r}^*}^2 = \inprod{x_{i_1}^*,x_{i_2 \hdots,i_r}^*} \geq \inprod{x_{i_1}^*,x_{u,i_2 \hdots,i_r}^*} = \norm{x_{u,i_1, \hdots,i_r}^*}^2  >0
\end{align*}
where the equality holds by consistency constraints, the first inequality above holds by constraint \prettyref{eq:subset_constraint} and the last inequality holds by equation \prettyref{eq:hyperedge_term_bound}. Hence for any $r$-tuple $\set{i_1,i_2,\hdots,i_r} \subseteq \mathcal{B}_u\paren{1,1 - \dfrac{1}{2r},V}$, we have $\norm{x^{*}_{i_1,i_2,\hdots,i_r}}^2 > 0$. Therefore by SDP constraint \prettyref{eq:sdp2}, it cannot form a hyperedge. Hence, the set of vertices in $\mathcal{B}_u\paren{1,1 - \dfrac{1}{2r},V}$ is an independent set.
\end{proof}

\begin{definition}
	Let $\mathcal{S}_u$ denote the set of vertices formed by the union of all vertices by reading off the indices from the tuples of the set, $\mathcal{B}_u(l,r,V)$.
\end{definition}

Now, we have all the ingredients to prove our main result. We present the complete algorithm below and the proof of \prettyref{thm:main2}.
\\
\noindent
\RestyleAlgo{ruled}
\begin{minipage}{\linewidth}
	\begin{algorithm}[H]
		\caption{}
		\label{alg:one}
		\begin{algorithmic}[1]
			\REQUIRE $H=(V, E)$, $l \in [r]$, and $\mathcal{R} \in (0, 1)$.
			\ENSURE A list of independent sets in $H$.
			\STATE Solve \prettyref{sdp:his}.
			\FORALL{$u \in V$}
			\STATE Initialize $\mathcal{S}_u$ denote the union of set of vertices from the tuples in $\mathcal{B}_{u}(l, \mathcal{R}, V)$.
			\STATE $\mathcal{S}_u' = \set{u} \cup \mathcal{S}_u$. If $\mathcal{S}_u'$ is not an independent set,\\ Set $\mathcal{S}_u' = \emptyset$ and skip this iteration.
			\FORALL{$v \in V \setminus \mathcal{S}_u$}
			\STATE Add vertex $v$ to $\mathcal{S}_u'$ if $\mathcal{S}_u' \cup \set{v}$ is an independent set. \label{step:greedy}
			\ENDFOR
			\ENDFOR
			\STATE Return $\set{\mathcal{S}_u'}_{u \in V}$.
		\end{algorithmic}
	\end{algorithm}
\end{minipage}

We set our parameters ($n,p,k,\e$) appropriately and show that the number of vertices in $\mathcal{B}_u$ along with the vertex $u$ (denoted by $\mathcal{S}'_u$) cover $1 - \e$ fraction of vertices in $S$.

\begin{theorem}[Formal version of \prettyref{thm:main2}]
	\label{thm:main2_formal}
	There exists a deterministic algorithm which takes as input $\e \in (0,1)$ and	
	an instance of \prettyref{def:model} satisfying 
 	\[
		k \geq \max\set{\dfrac{r2^{2r+2}e^{r}}{3p}, \dfrac{(2rf(r))^{1/(r-0.5)}n^{(r-1)/(r-0.5)}}{{\e^{1/(r-0.5)}}p^{1/(2r-1)}}} ,
	\]
	has running time $n^{\bigo{r}}$, and outputs an independent set of size at least $(1 - \e)k$, with high probability (over the randomness of the input).
\end{theorem}

\begin{proof}
\label{proof:thm_main2}
We run the Algorithm \ref{alg:one} with the inputs, $H,l=1$ and $\mathcal{R} = 1 - \dfrac{1}{2r}$ to get $\set{\mathcal{S}_u'}_{u \in V}$.
In \prettyref{lem:5.5} we show that
\[ \abs{\mathcal{B}_u\paren{1,1 - \dfrac{1}{2r}, S}} \geq 
	(k-1)\paren{1 - \dfrac{2rf(r)n^{r-1}}{\sqrt{p}k^{r-0.5}}}\mper \]
For a suitable choice of parameters we wish to have,
\begin{equation}
	\label{eq:set_params}
	\abs{\mathcal{B}_u\paren{1,1 - \dfrac{1}{2r},S}} \geq (k-1)(1-\e) \mper
\end{equation}
We can then  include the vertex $u$ to our independent set and we get
\begin{align*}
\abs{\mathcal{S}_u'}
&\geq \abs{\mathcal{S}_u} + 1
= \abs{\mathcal{B}_u\paren{1,1 - \dfrac{1}{2r},V}} + 1
\geq \abs{\mathcal{B}_u\paren{1,1 - \dfrac{1}{2r},S}} + 1\\
&\geq (k-1)(1-\e) + 1 \geq k(1-\e) \mper
\end{align*}
We note that by setting $k \geq \dfrac{(2rf(r))^{1/(r-0.5)}n^{(r-1)/(r-0.5)}}{{\e^{1/(r-0.5)}}p^{1/(2r-1)}}$, equation \prettyref{eq:set_params} is satisfied and hence we can recover an independent set of size $(1-\e)k$ for all $\e \in (0,1)$.
\end{proof}

\section{Algorithm for Exact Recovery of $S$}
\label{sec:full_recovery}

In this section, we will prove a formal version of \prettyref{thm:main} which is a generalization of Theorem 1.2 of \cite{McKenzie_2020} to $r$-uniform hypergraphs. We start by rewriting the lower bound on the SDP mass from the vectors in S, i.e.,  \prettyref{cor:lower_bound_S} into a form which is easier to work with.

Note that the \prettyref{lem:5.3}, tells us the that the projection of $(r-1)^{th}$ level vectors from $S$ have a large projection (close to 1), onto some vertex $u \in S$. This naturally suggests that we iterate over all vertices and consider the $(r-1)^{th}$ level vectors which have a large projection with the vertex. To ensure that union of such projected sets remain independent, we generalize the ideas in proof of Theorem 1.1 and Theorem 1.2 of \cite{McKenzie_2020} to higher level vectors and hyperedges (\prettyref{lem:4.4}, \prettyref{lem:4.5}, \prettyref{lem:4.6} and proof of \prettyref{thm:main}). We start by the following simple yet important lemma, where we show that there exists a constant value of $\mathcal{R}$ such that no two orthogonal vectors can belong to $\mathcal{B}_{u}(r-1,\mathcal{R}, T)$ for any $u \in V, T \subseteq V$.

\begin{lemma}\label{lem:4.4}
	Let $w$ be a fixed unit vector. Then for all $\mathcal{R} > \ffrac{1}{\sqrt{2}}$ and for any vector $y$ which satisfies $\norm{y} \leq 1$ and $\inprod{w, y} \geq \mathcal{R}$, every vector $z$ such that $\norm{z} \leq 1 \text{ and } \inprod{y, z} = 0$ must have $\inprod{w, z} < \mathcal{R}$.
\end{lemma}

\begin{proof}
	Let us suppose there is a vector $z$ which on the contrary does satisfy $\inprod{w, z} \geq \mathcal{R}$. First we decompose $y$ and $z$ as follows,
	\[	y = \inprod{w, y} w + y_{\perp} \qquad \textrm{and} \qquad
	z = \inprod{w, z} w + z_{\perp}
	\]
	where $\inprod{w, y_{\perp}} = \inprod{w, z_{\perp}} = 0$. Since $\norm{y}, \norm{z} \leq 1$ and $\norm{w} = 1$, we have
	\[ \norm{y_{\perp}} \leq \sqrt{1-\inprod{w, y}^2} \leq \sqrt{1-\mathcal{R}^2} \qquad \textrm{and} \qquad
	\norm{z_{\perp}} \leq \sqrt{1-\inprod{w, z}^2} \leq \sqrt{1-\mathcal{R}^2}\mper
	\]
	Using $0 = \inprod{y, z} = \inprod{w, y} \inprod{w, z} + \inprod{y_{\perp}, z_{\perp}}$,
	\[
	\inprod{w, z} = \dfrac{-\inprod{y_{\perp}, z_{\perp}}}{\inprod{w, y}} \leq \dfrac{{\norm{y_{\perp}}}\norm{z_{\perp}}}{\inprod{w, y}} \leq \dfrac{1-\mathcal{R}^2}{\mathcal{R}}\mper
	\]
	But note that for any $\mathcal{R} \in (\ffrac{1}{\sqrt{2}}, 1)$,
	\[
	\dfrac{1-\mathcal{R}^2}{\mathcal{R}} < \mathcal{R}\mper
	\]
	This is a contradiction to the fact that $\inprod{w, z} \geq \mathcal{R}$.
\end{proof}

For the rest of our discussion we pick $\mathcal{R} = 3/4$. We note that the value the $\mathcal{R} = 3/4$ is an arbitrary choice for the constant and is used for the purposes of presentation only and has no particular significance.

Thus \prettyref{lem:4.4} ensures that no two orthogonal vectors can lie in $\mathcal{B}_u(r-1,3/4,T)$ for any $u \in S, T \subseteq V$. Next, we give a lower bound on the size of $\mathcal{B}_u(r-1,3/4,S)$, which will tell us that a large number of the tuples lie within this set.

\begin{lemma}\label{lem:4.5}
	For $k \geq \dfrac{r2^{2r+2}e^{r}}{3p}$, there exists a vertex $u \in S$ such that,
	\[ \Abs{\mathcal{B}_u(r-1,3/4,S)} \geq {{k-1} \choose {r-1}} \paren{1-\dfrac{4f(r)n^{r-1}}{k^{r-0.5}\sqrt{p}}}\mper \]
\end{lemma}

\begin{proof}
	By Markov's inequality for all $\mathcal{R} \in (0, 1)$ and where  $k \geq \dfrac{r2^{2r+3}e^{r}}{3p}$, by \prettyref{lem:5.3} there exists a vertex $u \in S$ such that,
	\begin{align*}
		\mathbbm{P}_{\set{i_1,i_2,\hdots,i_{r-1}} \sim {S \setminus \set{u} \choose r-1}}\left[1-\inprod{x^{*}_u,x^{*}_{i_1,i_2,\hdots,i_{r-1}}} > 1 - \mathcal{R}\right] &< \dfrac{1-\mathbbm{E}_{\set{i_1,i_2,\hdots,i_{r-1}} \sim S \setminus \set{u}}\inprod{x^{*}_u,x^{*}_{i_1,i_2,\hdots,i_{r-1}}}}{1-\mathcal{R}}\\
		&< \dfrac{f(r)n^{r-1}}{\paren{1-\mathcal{R}}k^{r-0.5}\sqrt{p}}\mper
	\end{align*}
We can rewrite the above expression as the fraction of $r-1$ tuples which satisfy  ${\inprod{x^{*}_u,x^{*}_{i_1,i_2,\hdots,i_{r-1}}} < \mathcal{R}}$, since the underlying distribution is the uniform distribution over all such tuples,
	\begin{align*}
		\abs{\set{i_1,i_2,\hdots,i_{r-1}} \in {S \setminus \set{u} \choose r-1} : \inprod{x^{*}_u,x^{*}_{i_1,i_2,\hdots,i_{r-1}}} < \mathcal{R} } &< {{k-1} \choose {r-1}} \paren{\dfrac{f(r)n^{r-1}}{\paren{1-\mathcal{R}}k^{r-0.5}\sqrt{p}}}.\\
		\therefore \abs{\set{i_1,i_2,\hdots,i_{r-1}} \in {S \setminus \set{u} \choose r-1} : \inprod{x^{*}_u,x^{*}_{i_1,i_2,\hdots,i_{r-1}}} \geq \mathcal{R} } &\geq {{k-1} \choose {r-1}} \paren{1-\dfrac{f(r)n^{r-1}}{\paren{1-\mathcal{R}}k^{r-0.5}\sqrt{p}}}\mper
	\end{align*}
Using $\mathcal{R} = 3/4$ finishes the proof of this claim.
\end{proof}
 
Recall that in \prettyref{thm:main}, we are aiming to recover $S$ and not just any independent set. For this, we also need that every vertex $v \in V \setminus S$, has at least one hyperedge forming with $\mathcal{B}_u(r-1,3/4,S)$. This will ensure that the set $\mathcal{B}_u(r-1,3/4,V)$ has tuples only from ${{S \setminus \set{u}} \choose r-1}$; we did show that this has a large size in \prettyref{lem:4.5}. We concretize these ideas in the next few lemmas.

\begin{lemma}\label{lem:4.6}
	For $k \geq \max\set{\dfrac{r2^{2r+2}e^{r}}{3p}, \dfrac{\paren{8f(r)}^{\flatfrac{1}{(r-0.5)}}n^{(r-1)/(r-0.5)}}{p^{3/(2r-1)}}, (r-1)\paren{\dfrac{16 \log n }{p}}^{\ffrac{1}{(r-1)}}}$ there exists a vertex $u \in S$ where $\forall v \in V \setminus S, \exists \; e \: =\paren{ \set{v_1,v_2,\hdots,v_{r-1}}\cup \set{v}} \in E $ such that $\set{v_1,v_2,\hdots,v_{r-1}} \in \mathcal{B}_u(r-1, 3/4,S)$ with high probability (over the randomness of the input).
\end{lemma}

\begin{proof}
As stated above, we are interested in the event that for each vertex in $v \in V \setminus S$, we have at least one hyperedge with $\mathcal{B}_u(r-1,3/4,S)$; where $u$ is the vertex guaranteed by \prettyref{lem:4.5}. This is implied by the event that every such $v$ participates in strictly more than ${k-1 \choose r-1} - \abs{\mathcal{B}_u\paren{r-1,3/4,S}}$ hyperedges in ${S \choose {r-1}}$. We want that,
\[ d_{\min} \defeq \min_{v \in V \setminus S} d(v)|_{S \choose {r-1}} > {k-1 \choose r-1} - \abs{\mathcal{B}_u\paren{r-1,3/4,S}} \mper \] which for  $k \geq \dfrac{r2^{2r+3}e^{r}}{3p}$, using \prettyref{lem:4.5} is implied by the event that
\begin{equation}
	\label{eq:min_degree}
	d_{\min} > {k-1 \choose r-1} - {{k-1} \choose {r-1}} \paren{1-\dfrac{4f(r)n^{r-1}}{k^{r-0.5}\sqrt{p}}} = {{k-1} \choose {r-1}} \paren{\dfrac{4f(r)n^{r-1}}{k^{r-0.5}\sqrt{p}}}	\mper
\end{equation}
Note that for $v \in V \setminus S,$ $\mathbbm{E}[d(v)|_{S \choose {r-1}}] = p{k-1 \choose r-1}$, and since expectation is always larger than the minimum,  we must have that,
\begin{align}
	\label{eq:lower_bound_p}
	p{k-1 \choose r-1} >  d_{min} > {{k-1} \choose {r-1}} \paren{\dfrac{4f(r)n^{r-1}}{k^{r-0.5}\sqrt{p}}}
	\iff p > \dfrac{4f(r)n^{r-1}}{k^{r-0.5}\sqrt{p}}\mper
\end{align}
which is true since  $k \geq \dfrac{\paren{8f(r)}^{\flatfrac{1}{(r-0.5)}}n^{(r-1)/(r-0.5)}}{p^{3/(2r-1)}}$.
To show that this event happens with high probability, we use the Chernoff bound (\prettyref{fact:chernoff}) followed by a union bound over all $v \in V \setminus S$ to get,
\begin{align}
\mathbb{P}\left[\exists v \in V \setminus S : d(v)|_{S \choose {r-1}} \leq(1-\e)p{{k-1} \choose {r-1}}\right]
	& \leq (n-k)\exp\paren{\frac{-\e^2{(1/(r-1))^{r-1}}pk^{r-1}}{2}} \nonumber \\
 & \leq n\exp\paren{\frac{-\e^2{(1/(r-1))^{r-1}}pk^{r-1}}{2}}  \label{eq:degree_lower_bound}
\end{align}
\noindent 
where we used \prettyref{fact:sterling} and the fact that $r \geq 2$. The bound in equation  \prettyref{eq:degree_lower_bound} holds with high probability for $\e=1/2$ and $k \geq (r-1) \paren{\dfrac{16 \log n }{p}}^{\ffrac{1}{(r-1)}}\mper$

\noindent
Therefore, for $k \geq \max\set{\dfrac{r2^{2r+2}e^{r}}{3p}, \dfrac{\paren{8f(r)}^{\flatfrac{1}{(r-0.5)}}n^{(r-1)/(r-0.5)}}{p^{3/(2r-1)}}, (r-1)\paren{\dfrac{16 \log n }{p}}^{\ffrac{1}{(r-1)}}}$ and $\e = 1/2$, each vertex in $v \in V \setminus S$ has at least one hyperedge with some tuple in $\mathcal{B}_u(r-1, 3/4, S)$, with high probability.
\end{proof}

We note that these results hold true even if allow a monotone adversary to add hyperedges in $\partial(S)$ since it can only increase $d(v)$.

\begin{Claim}
	\label{lem:new_constraint_4}
	Let $I, J, I', J'$ be any non-empty subsets of $V$ which satisfy $\abs{I \cup J}, \abs{I' \cup J'} \leq r+1, \text{ and } I \cup J = I' \cup J'$, then,
	\[
		\inprod{x_I^*,x_J^*} = \inprod{x_{I'}^*,x_{J'}^*} \mper
	\]
\end{Claim}

\begin{proof}\label{app:new_constraint_4_proof}
	By SDP constraint \ref{eq:sdp3},
	\[
		\inprod{x_I^*,x_J^*} = \norm{x_{I \cup J}^*}^2 = \norm{x_{I' \cup J'}^*}^2 = \inprod{x_{I'}^*,x_{J'}^*} \mper
	\]
\end{proof}

\begin{lemma}
	\label{lem:4.7}
For $k \geq \max\set{\dfrac{r2^{2r+2}e^{r}}{3p}, \dfrac{\paren{8f(r)}^{\flatfrac{1}{(r-0.5)}}n^{(r-1)/(r-0.5)}}{p^{3/(2r-1)}}, (r-1)\paren{\dfrac{16 \log n }{p}}^{\ffrac{1}{(r-1)}}}$, there exists a vertex $u \in S$ (the vertex guaranteed by \prettyref{lem:4.6}) such that $\mathcal{S}_u \subseteq S$  and $\mathcal{B}_u(r-1,3/4,V) = \mathcal{B}_u(r-1,3/4,S)$ with high probability (over the randomness of the input).
\end{lemma}

\begin{proof}
Let $u$ be the vertex guaranteed by \prettyref{lem:4.6}.
We show that for any $(r-1)$-tuple of vertices in $\mathcal{B}_u(r-1, 3/4,V)$, each of these vertices lie in $S$. Suppose there is a $(r-1)$-tuple $I = \set{i_1,i_2,\hdots,i_{r-1}} \not\subset S$ such that $i_1 \in V \setminus S$ and $I \in \mathcal{B}_u(r-1, 3/4,V)$. Now from \prettyref{lem:4.6}, we know that with high probability for this vertex $u \in S$ there exists a hyperedge $e = \set{i_1, j_2, j_3, \hdots, j_{r}}$ such that $J = \set{j_2, j_3, \hdots, j_{r}} \in \mathcal{B}_u(r-1, 3/4,S)$. By \prettyref{lem:new_constraint_4} we have that,
\[\inprod{x_I^*,x_{J}^*} = \inprod{x_{I \setminus \set{i_1}}^*,x_{J \cup \set{i_1}}^*} = \inprod{x_{I \setminus \set{i_1}}^*,x_{e}^*} \mper\]

We know that $\norm{x_{e}^*} = 0$ and thus $x_{I}^*$ is orthogonal to $x_J^*$. Thus by \prettyref{lem:4.4} and our choice of $\mathcal{R}~(=3/4)$, the set $I$ cannot lie inside $\mathcal{B}_u(r-1,3/4,V)$. Thus the only tuples that can lie inside $\mathcal{B}_u(r-1,3/4,V)$ are the tuples of $\mathcal{B}_u(r-1,3/4,S)$. Hence with high probability (over the randomness of the input), $\mathcal{B}_u(r-1,3/4,V) = \mathcal{B}_u(r-1,3/4,S)$. 
Since the vertices in $\mathcal{S}_u$ are formed by taking the union over vertices in the tuples inside $\mathcal{B}_u(r-1,3/4,V)$ we only have vertices from $S$. Therefore, $S_u \subseteq S$ is an independent set.
\end{proof}

\begin{theorem}[Formal version of \prettyref{thm:main}]
	\label{thm:main_formal}
	There exists a deterministic algorithm which takes as input 	
	an instance of \prettyref{def:model} satisfying
	\[ k \geq \max\set{\dfrac{r2^{2r+2}e^{r}}{3p}, \dfrac{\paren{8f(r)}^{\flatfrac{1}{(r-0.5)}}n^{(r-1)/(r-0.5)}}{p^{3/(2r-1)}}, (r-1)\paren{\dfrac{16 \log n }{p}}^{\ffrac{1}{(r-1)}}}, \]
	has running time $n^{\bigo{r}}$, and outputs a list of atmost $n$ independent sets, one of which is $S$,  with high probability (over the randomness of the input).
\end{theorem}

\begin{proof}
	We run the Algorithm \ref{alg:one} with the inputs, $H$, $l = r-1$, and $\mathcal{R} = \flatfrac{3}{4}$ to get $\set{\mathcal{S}_u'}_{u \in V}$. \prettyref{lem:4.7} guarantees that there exists a vertex $u$ such that ${\mathcal{S}_u}$ is a subset of our planted independent set $S$ with high probability (over the randomness of the input). The greedy step (step \ref{step:greedy}) in Algorithm \ref{alg:one}, which tries to add each vertex from $V \setminus \mathcal{S}_u$ to $\mathcal{S}_u'$, helps us to recover the remaining vertices of our planted independent set $S$. Note that by \prettyref{lem:4.6}, no vertex of $V \setminus S$ can be added to $\mathcal{S}_u'$ since for any vertex $v \in V \setminus S$, there exists a hyperedge containing $v$ and a subset of vertices from $\mathcal{S}_u$. Hence, at the end of greedy step we will completely recover $S$. Since the time to solve the SDP is $n^{\bigo{r}}$, the running time is $n^{\bigo{r}}$.
\end{proof}

\paragraph{Acknowledgements.}
YK and RP thank Theo McKenzie for helpful discussions. AL was supported in part by SERB Award ECR/2017/003296 and a Pratiksha Trust Young Investigator Award. 

\bibliographystyle{alpha}
\bibliography{bibfile}

\appendix
\section{Some standard inequalities}
\label{app:inequalitites}

\begin{fact}[Bounds on Binomial Coefficient, Appendix C - \cite{10.5555/1614191}]
	\label{fact:sterling}
	For $1 \leq k \leq n,$
	\[ \paren{\dfrac{n}{k}}^k \leq {n \choose k} \leq \paren{\dfrac{en}{k}}^k \mper\]
\end{fact}

\begin{fact}[Chernoff bound (Multiplicative); Theorem 4.5 (Part 2) - \cite{mitzenmacher2017probability}]
	\label{fact:chernoff}
	Let $X_1, X_2, \hdots, X_n$ be i.i.d. bernoulli variables such that $\mu = \E[X_i]$, for all $i$. Then for any $\delta \in (0, 1)$,
	\[
	\ProbOp\left[{\sum_{i=1}^{n} X_i} < (1-\delta) \mu \right] \leq \exp\paren{-\dfrac{\mu\delta^2}{2}}\mper
	\]
\end{fact}

\begin{fact}[Bernstein inequality; Equation (1.29) - \cite{wasserman2006all}]
	\label{fact:bern}
	Let $X_1, X_2, \hdots, X_n$ be independent, zero mean random variables such that $\abs{X_i} \leq c$, for all $i$. Then for all $t > 0$,
	\[
	\ProbOp\left[\abs{\sum_{i=1}^{n} X_i} > t \right] \leq 2\exp\paren{-\dfrac{t^2}{2v + \ffrac{2ct}{3}}}
	\]
	where $v \geq \sum_{i=1}^{n} \var{X_i}$.
\end{fact}

\begin{fact}[Grothendieck's inequality; Equation (3.1) and (3.2) - \cite{gudon2014community}]
	\label{fact:groth}
	Consider an $n \times n$ real matrix $M$ then,
	\[
	\max_{\set{x_i}_{i=1}^{n}, \set{y_j}_{j=1}^{n} \subseteq B_{2}^{n}} \abs{\sum_{i, j = 1}^{n} M_{ij} \inprod{x_i, y_j}} \leq 2 \max_{\set{\alpha_i}_{i=1}^{n}, \set{\beta_j}_{j=1}^{n} \in \set{-1, 1}} \abs{\sum_{i, j = 1}^{n} M_{ij} \alpha_i \beta_j}
	\]
	where $B_{2}^{n} = \set{x \in \R^n : \norm{x}_2 \leq 1}$.
\end{fact}

\section{Proof of feasibility of $x'$}\label{app:feasibility_proof}
Here we prove that the solution $x'$ is feasible for \prettyref{sdp:his} as stated in equation \prettyref{eq:feasible}.

\begin{enumerate}
\item
First we verify that $\norm{x'_i}^2 = 1, \; \forall  i \in V$
\begin{align*}
&\norm{x'_i}^2 = \norm{\hat{e}}^2 = 1, \; \forall i \in S  & \paren{\text{as defined}}\\
&\norm{x'_i}^2 = \norm{x^*}^2 = 1, \; \forall i \in V \setminus S &  \paren{x^* \text{ is also feasible}}
\end{align*}

\item
The constraint $\norm{x'_{e}}^2 = 0$ holds by construction since for any edge $e = \set{i_1,i_2,\hdots,i_r}$ such that $\set{i_1,i_2,\hdots,i_r} \cap S \neq \emptyset \neq \set{i_1,i_2,\hdots,i_r} \cap (V \setminus S)$ the feasible solution  $x'$ sets $x'_{i_1,i_2,\hdots,i_r} = 0$. Similarly for any edge $e$ completely inside $V \setminus S$, $x'_e = x^*_e = 0$.

\item
Next, we verify the constraint $\inprod{x'_I,x'_J} = \norm{x'_{I \cup J}}^2$. 

For $I, J \in S$, 
\[\ x'_I  = x'_J = x'_{I \cup J} = \hat{e} \implies \inprod{x'_I,x'_J} = \norm{x'_{I \cup J}}^2\mper\]
For $ I, J \in V \setminus S$,
\[x'=x^* \implies \inprod{x'_I,x'_J} = \norm{x'_{I \cup J}}^2 \textrm{ since $x^*$ is feasible}.\]
If $I \in \partial{S}$ (resp. $J \in \partial{S}$) then $I \cup J \in \partial{S}$. Moreover, by construction,
 $x'_I = 0 $ (resp. $x'_J = 0$) and $\norm{x'_{I \cup J}} = 0 = \inprod{x'_I,x'_J}$.
If $I,J \notin \partial{S}$ but $I \cup J \in \partial{S}$, then w.l.o.g. assuming $I \in S$ and $J \in V \setminus S$,
\[\inprod{x'_I,x'_J} = \inprod{\hat{e},x^*_I} = 0  = \norm{x'_{I \cup J}}^2 . \]

\item
To verify constraint \prettyref{eq:subset_constraint}, we consider $x'_I,x'_J$ such that $I \subseteq J$. We note that for $J \subseteq S$ or $J \subseteq V \setminus S$ the constraint holds by our definition of $x'$. Also for $I \cap S \neq \phi \neq I \cap (V\setminus S)$, the term on both sides are 0 and the constraint holds with an equality. The only case that remains to consider is when $J \cap S \neq \phi \neq J \cap (V \setminus S)$ but either $I \subseteq S$ or $I \subseteq V \setminus S$. For this case, $\inprod{x'_u,x'_J} = 0$ and by the SDP constraint $\inprod{x_u,x_I} \geq 0$, this holds true as well.

\item
To verify constraint \prettyref{eq:union_bound}, consider $x'_u,x'_{v_1},\hdots,x'_{v_r} $ for any $ \set{u,v_1,\hdots,v_r}$. If $\set{u,v_1,\hdots,v_r} \subseteq S$, the constraint holds with equality (Both sides are 0). For $\set{u,v_1,\hdots,v_r}  \subseteq V \setminus S$, the constraint holds by definition of $x'$. If $u \in S \textrm{ and } \forall i,v_i \in V \setminus S$ or $u \in V \setminus S \textrm{ and } \forall i,v_i \in S$, the term on left is $1$ since $x'_{u, v_1,\hdots,v_r }=0$ and the expression on right is $r$. If $u \in S \textrm{ and } \set{v_1,\hdots,v_r} \cap S \neq \phi \neq V \setminus S \cap \set{v_1,\hdots,v_r} $, then the term on left is 0 and $\exists v_i : \inprod{x_u,x_{v_i}}=0$ and hence the expression on right is $\geq 1$ and constraint is satisfied. If $u \in V \setminus S \textrm{ and }\set{v_1,\hdots,v_r} \cap S \neq \phi \neq V \setminus S \cap \set{v_1,\hdots,v_r} $, then the term on left evaluates to $1$ since $x'_{u, v_1,\hdots,v_r }=0$ but by same argument as in previous case the expression on right is $\geq 1$. So the constraint holds in all cases.
\end{enumerate}

This completes the proof.

\end{document}